\documentclass[%
reprint,
 showkeys,
 amsmath,amssymb,
 aps,
 pra,
]{revtex4-1}
\usepackage{graphics,epsfig}
\usepackage[]{graphicx}
\usepackage{latexsym}
\usepackage{amsmath, amsthm, amsfonts, amssymb}
\usepackage{verbatim}
\usepackage{graphicx}
\usepackage{dcolumn}
\usepackage{bm}

\begin{document}

\preprint{APS/123-QED}

\newcommand{\nd}{\noindent}
\newtheorem{theo}{Theorem}[section]
\newtheorem{definition}[theo]{Definition}
\newtheorem{lem}[theo]{Lemma}
\newtheorem{prop}[theo]{Proposition}
\newtheorem{coro}[theo]{Corollary}
\newtheorem{exam}[theo]{Example}
\newtheorem{rema}[theo]{Remark}
\newtheorem{example}[theo]{Example}
\newtheorem{principle}[theo]{Principle}
\newcommand{\ninv}{\mathord{\sim}} 
\newtheorem{axiom}[theo]{Axiom}

\title{Convex politopes and quantum separability}

\author{F. Holik}
\affiliation{Departamento de Matem\'{a}tica - Ciclo B\'{a}sico Com\'{u}n\\
Universidad de Buenos Aires - Pabell\'{o}n III, Ciudad
Universitaria \\ Buenos Aires, Argentina\\
}%
\affiliation{
 Postdoctoral Fellow of CONICET\\
}%
\author{A. Plastino}
\affiliation{%
 National University La Plata
\& CONICET IFLP-CCT, C.C. 727 - 1900 La Plata, Argentina
}%

\date{\today}

\begin{abstract}
\noindent We advance a novel perspective of the entanglement issue
that appeals to  the Schlienz-Mahler measure [Phys. Rev. A
\textbf{52}, 4396 (1995)]. Related to it, we propose a criterium
based on the consideration of convex subsets of quantum states. This
criterium generalizes a property of product states to convex subsets
(of the set of quantum-states) that is able to uncover a new
geometrical property of the separability property.
\begin{description}
\item[PACS numbers]
\textbf{03.65.Ud}
\end{description}
\end{abstract}

\pacs{Valid PACS appear here}
\maketitle

\bibliography{pom}

\section{Introduction}

\nd   Schr\"{o}dinger stated, as everyone knows, that ``entanglement
is {\bf the} characteristic trait of quantum
mechanics"\cite{Schro-35,Schro-36,EPR}. Many years afterwards,
entanglement, although still rather a puzzling issue, is a subject
of immense attention, mostly because interest on its
characterization has more than foundational significance, it being a
powerful resource for quantum information processing that offers a
host of possible technological applications
\cite{QuantumKeyDistribution}. A suggestive assertion
\cite{Horodeki-2009,horodeki2007review} seemingly deserves
repetition: ``The fundamental question in quantum entanglement
theory is \emph{which states are entangled and which are not}".

\subsection{Abstract mathematical notions and entanglement}
\nd The geometric properties of entanglement are of paramount
importance (see\cite{Horodeki-2009}). In order to characterize it,
many mathematical strategies have been followed, that  range from
the application of algebraic tools, to group theory, differential
geometry, convex geometry, numerical simulations, etc. (see
\cite{Horodeki-2009,horodeki2007review,PlenioVirmani-2006}). Without
any doubt, the discovery of new mathematical structures underlying
the theoretical description of entanglement has provided insightful
answers to the problems of its characterization, manipulation and
quantification, as remarked in \cite{PlenioVirmani-2006}. Underlying
many of these approaches, one encounters once and again geometrical
properties of the quantum set of states and, in particular, those of
the set of separable states \cite{Werner}. For examples of
geometrical applications to the study of entanglement see, for
instance,
\cite{Kus-Zyczkowski,Geom-Leinaasetal,ClassicalTensorsI,ClassicalTensorsII,GeometryofGrabowski,Geometry2qbits,StellaOctangulla,Geometry-Hugston,EntanglementIllustrated,GometricAspectsofEntanglement}
and also \cite{zyczkowski1998} for an excellent overview.

\nd Since characterizing the geometry of entanglement is indeed a
fundamental task for physicists, {\it we propose here to appeal to a
very powerful abstract concept for guiding entanglement-research},
namely, {\it the convex set of quantum states (CSQS)}, which
exhibits fascinating geometrical properties \cite{zyczkowski1998}.
The CSQS  not only deserves  mathematical interest, but also sheds
light on the abstract and counterintuitive properties of
entanglement,  the difference between entangled and separable states
being a conspicuous example \cite{Werner}. In a different vein,
information needed to reformulate quantum mechanics is fully
contained in the geometrical properties of the quantum set of states
\cite{MielnikGQM,MielnikGQS,MielnikTF}. Summing up: geometrical
knowledge about these properties underlies most of the current
research-lines on entanglement and opens the door to the possibility
of exploring non-linear generalizations of quantum mechanics. See
also
\cite{SchillingAshtekar-1995,KibbleGeometrization,Scwinger,GeometryofCoherentStates,Volume-ZyczcowskiI,Volume-ZyczcowskiII,Geometry-Prugovecki}
for more examples of geometrical applications. It seems odd to
regard any piece of mathematics as too abstract for
entanglement-physicists.

\subsection{Our goal} \nd This work pretends to exhibit {\it unexplored} geometrical properties of separable
states and also present a novel separability criterium (SC) closely
linked to the Schlienz-Mahler (SM) entanglement measure
\cite{Ent95}. Our SC is formulated in geometrical-convexity terms
and is easily exportable to more general environments via the
so-called convex operational approach to physics.

\nd Now, the SM measure alluded to above constitutes a really
significant development, being used as a basis not only for
developing new ideas but also to establish separability criteria
(see for example \cite{altafini}, \cite{Bjork1}, \cite{Bjork2},
\cite{Bjork3}, and \cite{zhang}). Their authors (SM) focus
attention {\it on the difference between a given density matrix
and the product of its reduced states} $\rho^{A}\otimes\rho^{B}$.
We will use a suitable generalization of this difference in order
to establish a link between the convex sets of the compound system
and its subsystems, thereby developing  a new entanglement
criterium based on the convex structure of the set of quantum
states. A similar derivation can be made by recourse to a quantum
logical approach \cite{holik}. Our admittedly abstract criterium
can still shed some light on the geometrical properties of
separable states.

\vskip 3mm \nd In working with the convex structure of the quantum
set of states we will regard convex subsets of it as probability
spaces and take advantage of the fact that some of these subsets
can be fully recovered from the information contained on the
available states of the associated subsystems. Such is our leit
motif. Further, we will advance the notion of informational
invariance and deal with \emph{convex invariant subsets}. Our
proposal is based on the property that for every separable state
there exists a convex subset which contains it and is an
informational invariant. From such basic idea, our
entanglement-edifice will be built up. It is endowed with the
strength of possibly allowing one to study and classify
entanglement in higher dimensions, and even to multipartite
systems just because of its abstract nature.

\nd Matters are organized as follows. After some preliminaries
(which, though not essential for the rest of the article, may
serve as a conceptual and mathematical guide) in Sec. II, we
review in section \ref{s:convex set of states} some ideas of
\cite{Ent95} together with their consequences. In section
\ref{s:New Criteria} we show how to construct special functions
that allow us to develop a new separability criterium. In section
\ref{s:Implications} we discuss implications of this criterium and
indicate how the functions so developed can be used to generalize
product states to convex sets. In section VI
we condense some of our results in a more conceptual fashion and,
 finally, draw some conclusions.

\section{Preliminaries}\label{s:preliminaries}

\nd The mathematically savvy reader should skip this Section. Given
a composite-system formed of subsystems $A$ and $B$, a fundamental
characteristic of a product state, i.e., a state of the form

\begin{equation}
\rho_{Prod}=\rho^{A}\otimes\rho^{B}
\end{equation}

\noindent is that information of the whole state may be
reconstructed  from the simple sum of the information on the
states of the subsystems. The ``simple sum" is mathematically
represented by taking tensor products on the reduced states of the
subsystems. Thus the above statement may be expressed in
mathematical terms: taking partial traces and making tensor
products leave the state unchanged. But not every separable state
has this property; in general, a separable state will be of a
non-product kind, and the above informational relationship is no
longer true. No entangled state has this property. Thus, only
product states are \emph{invariant} in this sense. Product states
are fully recovered from the information contained in the states
of the subsystems (to be abbreviated as the ``reobtained''
property). We may call this property the \emph{informational
invariance}.

\nd We may also ask, and this is an unconventional viewpoint, for
the subsets of the convex set of states that exhibit the
recoverable property. An important example is the whole set of
separable states itself. It has -by definition- the property of
being fully recoverable by making tensor products of the complete
set of states of the subsystems and closing them by mixing
operation \cite{Werner}. In this sense we recover the
informational invariance property referred to above. Given the set
of available states of two systems, a physical operation is that
of  taking tensor products and then mix the pertinent states.
States obtained using these operations (together with local
unitary evolutions and classical communication) are classically
reproducible \cite{Werner}. \emph{In this work  we give a precise
mathematical formulation for set-notions of the kind exemplified
above, as well as a geometrical characterization of them.} The
ensuing mathematical notions will reveal  novel geometrical
structures which, in turn, make room  for a better
characterization of quantum states. \vskip 3mm

\nd We will denote sets of states with the informational invariance
property as \emph{convex separable subsets} (CSS) and will show that
for every separable state there exists a convex subset which
contains it and is an informational invariant (strictly included in
the convex set of separable states). Such indeed is the basis of our
abstract separability criterium, to be advanced below. Another
important feature of our abstract construction is the attainment of
a  purely geometrical description based on the convex structure of
the quantum set of states. The associated geometric reformulation of
entanglement may be useful for generalizing it to more general
scenarios, based on convex sets
\cite{Barnum-Wilce-2006,Barnum-Wilce-2009,Barnum-Wilce-2010}.

\subsection{Basic math-definitions}

\nd We remind the reader that every subset $A$ of a vector space is
contained within a smallest convex set called the convex hull of
$A$, namely the intersection of all convex sets containing $A$.
Thus, it is possible to define a convex-hull map $Conv()$ which has
three characteristic properties: i) extensivity  $A \subseteq
Conv(A)$, ii)   non-decreasing nature $A \subseteq B$ implies that
$Conv(A) \subseteq Conv(B)$, and iii) idempotency  $Conv(Conv(A)) =
Conv(A)$. Also, an extremal point of a convex set $S$ in a real
vector space is a point in $S$ which does not lie in any open line
segment joining two points of $S$ (an extremal point would be  a
``corner" of $S$). An important example for quantum mechanics is
that of pure states: they are the extreme points of the CSQS (more
on this below). \vskip 2mm

\nd A convex polytope may be defined as the convex hull of a finite
set of points (which are always bounded), or as a bounded
intersection of a finite set of half-spaces. One often asserts that
the term ``polytope'' is i) the general vocable of the sequence
``point, line segment, polygon, polyhedron, ...," or ii) to be
regarded  as a finite region of an $n-$dimensional space enclosed by
a finite number of hyperplanes. A $d-$dimensional polytope may be
specified as the set of solutions to a system of linear inequalities

\begin{equation}
M \bf{x} \le \bf{b},
\end{equation}

\noindent where $M$ is a real $s\times d$ matrix, and ${\bf b}$ is a
real $s-$vector. \vskip 3mm \nd For quantum systems,
$\mathcal{P}(\mathcal{H})$ will denote the set of all closed
subspaces of the pertinent Hilbert space $\mathcal{H}$, which are in
a one to one correspondence with the projection operators. Because
of this one to one link, one usually employs the notions of ``closed
subspace'' and ``projector'' in interchangeable fashion. An
important construct is $\mathcal{A}$, the set of bounded Hermitian
operators on $\mathcal{H}$, while the bounded operators on
$\mathcal{H}$ will be denoted by $\mathcal{B}(\mathcal{H})$. Pure
quantum states may be put in correspondence with the projective
space $\mathbf{C}\mathbf{P}(\mathcal{H})$ of a complex Hilbert space
$\mathcal{H}$, which is the set of equivalence classes of vectors
$v$ in $\mathcal{H}$, with $v \ne 0$, for the relation given by $v
\sim w$ when $v = \lambda w$ with $\lambda$ a non-zero scalar. Here
the equivalence classes for $\sim$ are also called projective rays.
A trace class operator is a compact one for which a  finite trace
may be defined (independently of the choice of basis).

\vskip 3mm \nd We will appeal below to the set $\mathcal{C}$
containing  all positive, hermitian, and trace-class (normalized to
unity) operators in $\mathcal{B}(\mathcal{H})$. \emph{A larger and
important structure used below, is the one denoted by
$\mathcal{L}_{\mathcal{C}}$, the set of all convex subsets of
$\mathcal{C}$.} This structure is  endowed with a lattice structure.
Finally, the reader may wish to recall in the Appendix some
elementary set-theory concepts used in the text. It is important to
remark that we will restrict to the finite dimensional case in the
rest of this work.

\begin{figure*}
\begin{center}
\includegraphics[width=8cm]{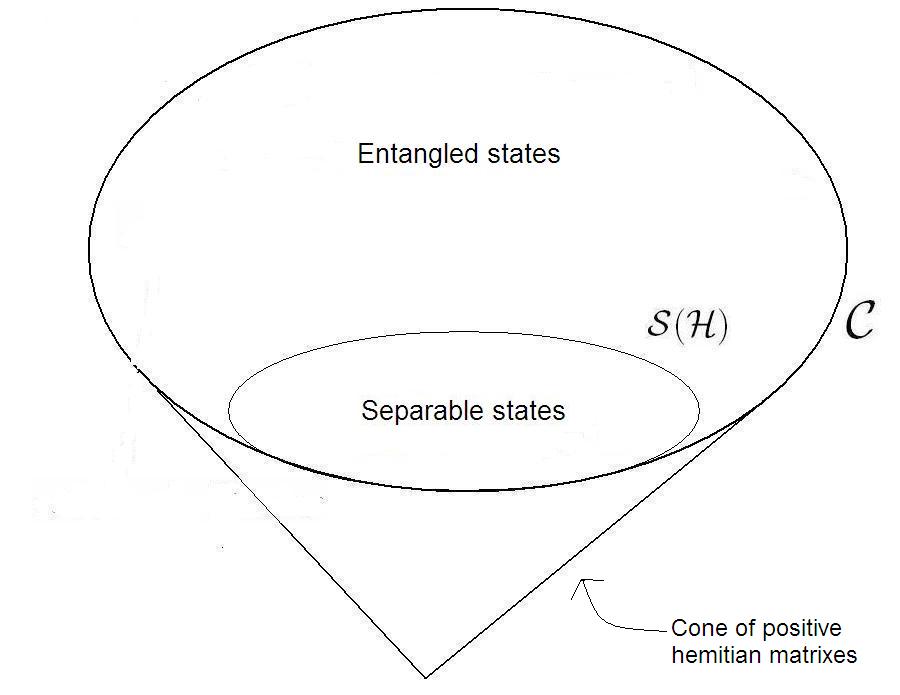}
\caption{\label{Graficoconvexo}\small{Geometric representation of
the convex set of states.}}
\end{center}
\end{figure*}



\section{The Schlienz-Mahler entanglement measure}\label{s:convex set of states}

\nd   For two quantum  systems $S_{1}$ and $S_{2}$, if
$\{|\varphi_{i}^{(1)}\rangle\}$ $-$ $\{|\varphi_{i}^{(2)}\rangle\}$
are the corresponding orthonormal basis of $\mathcal{H}_{1}$ $-$
$\mathcal{H}_{2}$, respectively, then the set
$\{|\varphi_{i}^{(1)}\rangle\otimes|\varphi_{j}^{(2)}\rangle\}$
constitutes an orthonormal basis for
$\mathcal{H}_{1}\otimes\mathcal{H}_{2}$. A general (pure) state of
the composite  $S_{1}-S_{2}$ system can be written as:

\begin{equation}
\rho=|\psi\rangle\langle\psi|.
\end{equation}

\noindent with $|\psi\rangle$ any vector in
$\mathcal{H}_{1}\otimes\mathcal{H}_{2}$. In the finite dimensional
case mixtures are represented by positive, Hermitian and trace one
operators (also called `density matrices'). The set of all density
matrixes forms a convex set (of states), which was called
$\mathcal{C}$ above, while the physical observables are represented
by elements of $\mathcal{A}$, the vector space of Hermitian
operators acting on $\mathcal{H}$. Formally we deal with the sets

\begin{definition}\label{d:hermitian}
$\mathcal{A}:=\{ A\in B(\mathcal{H})\,|\, A=A^{\dagger}\}$
\end{definition}

\begin{definition}\label{d:mathcalC}
$\mathcal{C}:=\{\rho\in\mathcal{A}\,|\,\mbox{tr}(\rho)=1,\,\rho\geq
0\},$
\end{definition}
\noindent where $B(\mathcal{H})$ stands for the algebra of bounded
operators in $\mathcal{H}$. $\mathcal{C}$ is a convex set inside the
hyperplane $\{\rho\in\mathcal{A}\,|\,\mbox{tr}(\rho)=1\}$ formed by
the intersection of this hyperplane with the cone of positive
matrices (see Figure \ref{Graficoconvexo}). Separable states are
defined \cite{Werner,zyczkowski1998} as those states of
$\mathcal{C}$ which can be written as a convex combination of
product states:

\begin{equation}
\rho_{Sep}=\sum_{i,j}\lambda_{ij}\rho_{i}^{(1)}\otimes\rho_{j}^{(2)},
\end{equation}
\noindent where $\rho_{i}^{(1)}\in\mathcal{C}_{1}$, and
$\rho_{j}^{(2)}\in\mathcal{C}_{2}$, $\sum_{i,j}\lambda_{ij}=1$ and
$\lambda_{ij}\geq 0$. We denote the set of separable states by
$\mathcal{S}(\mathcal{H})$.

\vskip 3mm \nd In set-parlance, the collective of  entangled states
becomes precisely  defined by

\begin{equation}
\mathcal{E}(\mathcal{H}):=\mathcal{C}\setminus\mathcal{S}(\mathcal{H}),
\end{equation}

\noindent where ``$\setminus$'' stands for set-theoretical
difference.

\vskip 3mm \nd As the dimension of the Hilbert space grows, most of
the states in $\mathcal{C}$ are non separable \cite{aubrum2006}. The
estimation of the volume of $\mathcal{S}(\mathcal{H})$ is of great
interest (see --among others--\cite{Volume-ZyczcowskiI},
\cite{aubrum2006} and \cite{horodecki2001}). The entanglement
measure advanced in \cite{Ent95} is based on the Fano decomposition
\cite{Fano83} (see also \cite{zyczkowski1998}, page 349). For
$\rho\in\mathcal{C}$, if the dimension of the Hilbert space is $d$,
one expresses it in terms of $\{\sigma_{i}\}$, the $d^{2}-1$
generators of $SU(d)$ (the group of special unitary matrixes acting
on $\mathcal{H}$). For composite bipartite systems, if $d=NK$, then
we have the following decomposition (in terms of the basis
$SU(N)\otimes SU(K)$)
\begin{widetext}
\begin{equation}\label{e:fano}
\rho=\frac{1}{NK}(\textbf{1}_{NK}+\sum_{i=1}^{N^{2}-1}\tau_{i}^{A}\sigma_{i}
\otimes\textbf{1}_{K}
+\sum_{j=1}^{K^{2}-1}\tau_{j}^{B}\textbf{1}_{N}
\otimes\sigma_{j}+\sum_{i=1}^{N^{2}-1}\sum_{j=1}^{K^{2}-1}\beta_{ij}\sigma_{i}\otimes\sigma_{j}),
\end{equation}
\end{widetext}

\noindent where $\tau_{i}^{A}$ and $\tau_{j}^{B}$ are Bloch vectors
such that

\begin{equation}\label{e:Blochvectors}
\rho^{A}=\frac{1}{N}(\textbf{1}_{N}+
\sum_{i=1}^{N^{2}-1}\tau_{i}^{A}\sigma_{i}),
\end{equation}

\noindent with an analogous form for $\rho^{B}$. $\rho^{A}$ and
$\rho^{B}$ are the reduced density matrixes of subsystems $A$ and
$B$ respectively. Schlienz-Mahler (SM) note that the term
$\sum_{i=1}^{N^{2}-1}\sum_{j=1}^{K^{2}-1}\beta_{ij}\sigma_{i}\otimes\sigma_{j}$
is related to correlations and proceed to construct an entanglement
measure using it. SM define then the tensor

\begin{equation}
M_{ij}=\beta_{ij}-\tau_{i}^{A}\tau_{j}^{B},
\end{equation}
that will play a leading role in their considerations. They use
$\mbox{tr}(\mathrm{M}\mathrm{M}^{\dag})$ as a measure of
entanglement (up to normalization), and this measure conveys
essentially the same information as

\begin{equation}\label{e:schlienz-mahler}
\|\rho-\rho^{A}\otimes\rho^{B}\|_{\mathcal{H}\mathcal{S}}^{2},
\end{equation}
\noindent where $\|\cdots\|_{\mathcal{H}\mathcal{S}}$ is the Hilbert
Schmidt norm

\begin{equation} \label{normita}
\|A\|_{\mathcal{H}\mathcal{S}}^{2}=\mbox{tr}(AA^{\dagger}),
\end{equation}
\noindent for any $A\in \mathcal{B}(\mathcal{H})$. The measure
(\ref{e:schlienz-mahler})

\begin{itemize}

\item vanishes for any product state,

\item is positive elsewhere,

\item it is maximal for any pure state with
vanishing Bloch vectors $\tau_{i}^{A}$ and $\tau_{j}^{B}$ (Equation
(\ref{e:Blochvectors})), and \item it is invariant under local
unitary transformations.

\end{itemize}

\nd Such properties allow for the development of other entanglement
measures and entanglement criteria (see, for example, \cite{zhang}).
The distance induced by the trace norm between two states represents
how well two states can be distinguished via measurement
\cite{Wootters}. It can be shown \cite{Bjork3} that

\begin{equation}
\sum_{i,j=1}^{3}C^{2}(\hat{\sigma}^{A}_{i},\hat{\sigma}^{B}_{j})=4\mbox{tr}[(\rho-\rho^{A}\otimes\rho^{B})^{2}],
\end{equation}
\noindent with

\begin{eqnarray}
&C(\sigma^{A}_{i},\sigma^{B}_{j})=&\nonumber\\
&\langle\sigma^{A}_{i}\otimes\sigma^{B}_{j}\rangle-
\langle\sigma^{A}_{i}\otimes\mathbf{1}^{B}\rangle\langle\mathbf{1}^{A}\otimes\sigma^{B}_{j}\rangle&,
\end{eqnarray}
\noindent making (\ref{e:schlienz-mahler}) easy to implement because
it can be measured via single-rates and coincidence-rates. More
generally, functions of the form

\begin{equation}\label{e:family}
W(\rho)=\|F(\rho-\rho^{A}\otimes\rho^{B})\|,
\end{equation}

\noindent have been studied  in some detail (see for example
\cite{Bjork1},\cite{Bjork2}, \cite{Bjork3} and \cite{zhang}).
$\|\cdots\|$ denotes a norm on the space of density matrixes and
$F:\mathcal{C}\longrightarrow\mathcal{C}$ a useful function for the
study of entanglement. Thus, entanglement measures
(\ref{e:schlienz-mahler}) become special cases of (\ref{e:family}).
The conditions imposed on $F$ and $\|\cdots\|$ are such that $W$
satisfies a similar set of conditions than the ones imposed on the
SM measure listed above.

\nd In the following section we  show that entanglement measures of
the form (\ref{e:family}) are closely linked to a particular
separability criterium that generalizes the map which assigns
$\rho^{A}\otimes\rho^{B}$ to any composite density matrix $\rho$.

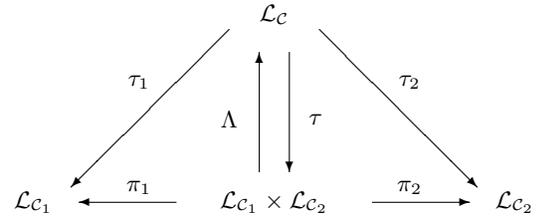
\begin{figure}
\begin{center}
\unitlength=1mm
\begin{picture}(5,5)(0,0)
\put(-6,23){\vector(-1,-1){21}} \put(6,23){\vector(1,-1){21}}
\put(-2,4){\vector(0,2){16}} \put(2,20){\vector(0,-2){16}}
\put(13,0){\vector(3,0){13}} \put(-13,0){\vector(-3,0){13}}

\put(0,25){\makebox(0,0){$\mathcal{L}_{\mathcal{C}}$}}
\put(-32,0){\makebox(0,0){$\mathcal{L}_{\mathcal{C}_{1}}$}}
\put(32,0){\makebox(0,0){$\mathcal{L}_{\mathcal{C}_{2}}$}}
\put(0,0){\makebox(0,0){$\mathcal{L}_{\mathcal{C}_{1}}\times\mathcal{L}_{\mathcal{C}_{2}}$}}
\put(-1,11){\makebox(-10,0){$\Lambda$}}
\put(-1,11){\makebox(13,0){$\tau$}}
\put(-18,16){\makebox(0,0){$\tau_{1}$}}
\put(18,16){\makebox(0,0){$\tau_{2}$}}
\put(-18,2){\makebox(0,0){$\pi_{1}$}}
\put(18,2){\makebox(0,0){$\pi_{2}$}}
\end{picture}
\caption{\label{f:figuremaps}The different maps between
$\mathcal{L}_{\mathcal{C}_{1}}$, $\mathcal{L}_{\mathcal{C}_{2}}$,
$\mathcal{L}_{\mathcal{C}_{1}}\times\mathcal{L}_{\mathcal{C}_{2}}$,
and $\mathcal{L}_{\mathcal{C}}$. $\pi_{1}$ and $\pi_{2}$ represent
canonical projections.}
\end{center}
\end{figure}

\section{A new separability criterium}\label{s:New Criteria}

\subsection{Preliminary matters}
\nd A glance at the Appendix might be useful at this stage. In the
previous section we saw how to construct a family of entanglement
measures via the mapping

\begin{definition}\label{e:assignment}
$$\Omega:\mathcal{C}\longrightarrow\mathcal{C}$$
$$\rho\mapsto \rho^{A}\otimes\rho^{B}.$$
\end{definition}

\nd Product states $\rho=\rho^{A}\otimes\rho^{B}$ satisfy

\begin{equation}\label{e:productproperty}
\Omega(\rho^{A}\otimes\rho^{B})=\rho^{A}\otimes\rho^{B},
\end{equation}

\noindent and they are the only states which satisfy
(\ref{e:productproperty}). Our leading idea now is that of {\it
generalizing the above considerations  to convex subsets of}
$\mathcal{C}$.

\subsubsection{First notion-generalization}

\nd In order to do so let us first study maps onto the set of states
of the subsystems, $\mathcal{C}_{1}$ and $\mathcal{C}_{2}$. We start
by defining special ``mappings'' using partial traces

\begin{eqnarray}
&\mbox{tr}_{i}:\mathcal{C}\longrightarrow \mathcal{C}_{j}&\nonumber\\
&\rho\mapsto \mbox{tr}_{i}(\rho)&,
\end{eqnarray}

\noindent from which we can construct the induced maps $\tau_i$ on
$\mathcal{L}_{\mathcal{C}}$, the set of all convex subsets of
$\mathcal{C}$ (a similar definition for
$\mathcal{L}_{\mathcal{C}_{i}}$, $i=1,2$), via the image of any
subset $C\subseteq\mathcal{C}$ under $\mbox{tr}_{i}$

\begin{eqnarray}
&\tau_{i}:\mathcal{L}_{\mathcal{C}}\longrightarrow
\mathcal{L}_{\mathcal{C}_{i}}&\nonumber\\
&C\mapsto \mbox{tr}_{j}(C)&,
\end{eqnarray}

\noindent where for $i=1$ we take the partial trace with $j=2$ and
vice versa. Thus, we can define the product map

\begin{eqnarray}
&\tau:\mathcal{L}_{\mathcal{C}}\longrightarrow\mathcal{L}_{\mathcal{C}_{1}}\times\mathcal{L}_{\mathcal{C}_{2}}&\nonumber\\
&C\mapsto(\tau_{1}(C),\tau_{2}(C))&
\end{eqnarray}

\noindent which generalizes partial traces to convex subsets of
$\mathcal{C}$.

\noindent In order to complete the desired generalization, let us
now define for convex subsets a new set-operation
$C_1\widetilde{\otimes} C_2$ that might be regarded as the analogous
of the tensor product (see Figure \ref{f:figuremaps}). We are thus,
loosely speaking, dealing with ``quasi-tensor set-compositions'' and
accordingly introduce the set of the definition that follows:

\begin{definition}\label{d:tensorconvex}
Given convex subsets $C_{1}\subseteq\mathcal{C}_{1}$ and
$C_{2}\subseteq\mathcal{C}_{2}$ we consider the set constructed
according to

\begin{equation}
C_1\widetilde{\otimes}
C_2:=\{\rho_{1}\otimes\rho_{2}\,|\,\rho_{1}\in C_1,\rho_{2}\in C_2\}
\end{equation}
\end{definition}

\noindent The symbol ``$\widetilde{\otimes}$" has a tilde in order
to avoid confusing it with the usual product of convex sets. Using
this, we define the map:

\begin{definition}\label{d:lambda}
$$\Lambda:\mathcal{L}_{\mathcal{C}_{1}}\times\mathcal{L}_{\mathcal{C}_{2}}\longrightarrow\mathcal{L}_{\mathcal{C}}$$
$$(C_{1},C_{2})\mapsto Conv(C_1\otimes C_2)$$
\end{definition}

\noindent where $Conv(\cdots)$ stands for \emph{convex hull} of a
given set. Applying $\Lambda$ to the particular case of the quantum
sets of states of the subsystems ($\mathcal{C}_{1}$ and
$\mathcal{C}_{2}$), one sees that Definitions \ref{d:tensorconvex}
and \ref{d:lambda} entail

\begin{equation}
\Lambda(\mathcal{C}_{1},\mathcal{C}_{2})=Conv(\mathcal{C}_{1}\widetilde{\otimes}\mathcal{C}_{2})
\end{equation}

\noindent and so, this is nothing but

\begin{equation}
\Lambda(\mathcal{C}_{1},\mathcal{C}_{2})=\mathcal{S}(\mathcal{H})
\end{equation}

\noindent because $\mathcal{S}(\mathcal{H})$ is by definition (for
finite dimension) the convex hull of the set of all product states
(which equals to
$\mathcal{C}_{1}\widetilde{\otimes}\mathcal{C}_{2}$). \emph{Thus,
the map $\Lambda$ gives a precise mathematical expression for the
operation of making tensor products and mixing mentioned in Section
\ref{s:preliminaries}}. Additionally, if
$\rho=\rho_{1}\otimes\rho_{2}$, with $\rho_{1}\in\mathcal{C}_{1}$
and $\rho_{2}\in\mathcal{C}_{2}$, then
$\{\rho\}=\Lambda(\{\rho_{1}\},\{\rho_{2}\})$, with
$\{\rho_1\}\in\mathcal{L}_{\mathcal{C}_{1}}$,
$\{\rho_2\}\in\mathcal{L}_{\mathcal{C}_{2}}$ and
$\{\rho\}\in\mathcal{L}_{\mathcal{C}}$. We can demonstrate as well
that

\begin{prop}
Let $\rho\in\mathcal{S(\mathcal{H})}$. Then, there exist
$C\in\mathcal{L}_{\mathcal{C}}$,
$C_{1}\in\mathcal{L}_{\mathcal{C}_{1}}$, and
$C_{2}\in\mathcal{L}_{\mathcal{C}_{2}}$ such that $\rho\in
C=\Lambda(C_{1},C_{2})$.
\end{prop}

\begin{proof}
If $\rho\in\mathcal{S(\mathcal{H})}$, then
$\rho=\sum_{ij}\lambda_{ij}\rho_{i}^{1}\otimes\rho_{j}^{2}$, with
$\sum_{ij}\lambda_{ij}=1$ and $\lambda_{ij}\geq 0$. Consider now the
convex sets

\begin{eqnarray}
C_{1}=Conv(\{\rho_{1}^{1},\rho_{2}^{1},\cdots,\rho_{k}^{1}\})\nonumber\\
C_{2}=Conv(\{\rho_{1}^{2},\rho_{2}^{2},\cdots,\rho_{l}^{2}\}).
\end{eqnarray}

\noindent We define:

\begin{equation}
C:=\Lambda(C_{1},C_{2})=Conv(C_{1}\otimes C_{2}).
\end{equation}

\noindent Clearly, the set
$\{\rho_{i}^{1}\otimes\rho_{j}^{2}\}\subseteq C_{1}\otimes C_{2}$,
and then $\rho\in C$.
\end{proof}

\subsubsection{Second notion-generalization}

\nd  The next notion to be tackled needs perhaps a perusal of
section II.A. We pass now to the  generalization  to convex subsets
of the map $\Omega$ in Definition \ref{e:assignment}. This is the
function $\Lambda\circ\tau$ (the composition of $\tau$ with
$\Lambda$). For the special case of a convex set formed by only one
``matrix'' (point) $\{\rho\}$ we have

\begin{equation}\label{e:lambdaenunrho}
\Lambda\circ\tau(\{\rho\})=\{\rho^{A}\otimes\rho^{B}\}
\end{equation}

\noindent which is completely equivalent to $\Omega$ and thus
satisfies (\ref{e:productproperty}). In what follows we will  need a
proposition taken from \cite{Convexsets}. It reads:

\begin{prop}
Let $S$ be a subset of a linear space $\mathcal{L}$. Then $x\in
Conv(S)$ iff $x$ is contained in a finite dimensional polytope
$\Delta$ whose extremal points belong to $S$,
\end{prop}

\nd    This is all we need to formulate now our proposal in the next
subsection.

\subsection{Our separability proposal}

\nd We will here ``traduce'' the idea of non separability as a
special kind of set-theory relationship.
\begin{prop}\label{subirbajar}
If $\rho$ is a separable state, then there exists a convex set
(indeed, a polytope), $S_{\rho}\subseteq\mathcal{S}(\mathcal{H})$
such that $\rho\in S_{\rho}$ and
$\Lambda\circ\tau(S_{\rho})=S_{\rho}$. More generally, for a
convex set $C\subseteq \mathcal{S}(\mathcal{H})$, there exists a
convex set $S_{C}\subseteq\mathcal{S}(\mathcal{H})$ such that
$\Lambda\circ\tau(S_{C})=S_{C}$. For a product state, we can
choose $S_{\rho}=\{\rho\}$. For any convex set
$C\subseteq\mathcal{C}$ which has at least one non-separable state
it is true  that there is no convex set $S$ such that $C\subseteq
S$ and $\Lambda\circ\tau(S)=S$.
\end{prop}

\begin{proof}
We have already seen above that if $\rho$ is a product state, then
$\Lambda\circ\tau(\{\rho\})=\{\rho\}$ and thus
$S_{\rho}=\{\rho\}$. If $\rho$ is a general separable state, then
there exists $\rho_{k}^{1}\in\mathcal{C}_{1}$,
$\rho_{k}^{2}\in\mathcal{C}_{2}$ and $\alpha_{k}\geq 0,
\sum_{k=1}^{N}\alpha_{k}=1$ such that
$\rho=\sum_{k=1}^{N}\alpha_{k}\rho_{k}^{1}\otimes\rho_{k}^{2}$.
Now consider the convex set (a polytope)

\begin{eqnarray}
M=\{\sigma\in\mathcal{C}\,|\,\sigma=\sum_{i,j=1}^{N}\lambda_{ij}\rho_{i}^{1}\otimes\rho_{j}^{2},\nonumber\\
\lambda_{ij}\geq 0, \sum_{i,j=1}^{N}\lambda_{ij}=1\}
\end{eqnarray}

\noindent $M$ contains all convex combinations of products of the
elements which appear in the decomposition of $\rho$. It should be
clear that $\rho\in M$. Let us compute $\Lambda\circ\tau(M)$, with
$\tau(M)=(\tau_{1}(M);\tau_{2}(M))$. An element of $\tau_{1}(M)$ is
of the form (for $\sigma\in M$)

\begin{equation}\label{e:cuenta}
\mbox{tr}_{1}(\sigma)=\sum_{i=1}^{N}(\sum_{j=1}^{N}\lambda_{ij})
\rho_{i}^{1}=\sum_{i=1}^{N}\mu_{i} \rho_{i}^{1},
\end{equation}

\noindent with $\mu_{i}=\sum_{j=1}^{N}\lambda_{ij}$. In analogous
fashion we show that an element of $\tau_{2}(M)$ is of the form
$\sum_{j=1}^{N}\nu_{j} \rho_{j}^{2}$ with
$\nu_{i}=\sum_{i=1}^{N}\lambda_{i,j}$. Note that
$\sum_{j=1}^{N}\mu_{j}=\sum_{j=1}^{N}\nu_{j}=1$. In order to compute
$\Lambda(\tau_{1}(M);\tau_{2}(M))$ we must build the convex hull of
the set

\begin{widetext}
\begin{equation}
\tau_{1}(M)\widetilde{\otimes}\tau_{2}(M)=\{\sigma_{1}\otimes\sigma_{2}|\sigma_{1}\in\tau_{1}(M),\sigma_{2}\in\tau_{2}(M)\}
=\{\sum_{i,j=1}^{N}\mu_{i}\nu_{j}\rho_{i}^{1}\otimes\rho_{j}^{2}\}.
\end{equation}
\end{widetext}

\noindent and  we conclude that

\begin{equation}\label{e:equality}
\Lambda\circ\tau(M)=Conv(\{\sum_{i,j=1}^{N}\mu_{i}\nu_{j}\rho_{i}^{1}\otimes\rho_{j}^{2}\}).
\end{equation}
Let us  prove that $\Lambda\circ\tau(M)=M$. If
$\sigma\in\Lambda\circ\tau(M)$, by looking at equation
(\ref{e:equality}) it is apparent that $\sigma$ belongs to $M$. On
the other hand, if $\sigma\in M$, then
$\sigma=\sum_{i,j=1}^{N}\lambda_{i,j}\rho_{i}^{1}\otimes\rho_{j}^{2}$
(convex combination). Note that $\Lambda\circ\tau(M)$ is a convex
set because trace operators preserve convexity and $\Lambda$ is a
convex hull. On the other hand,
$\Lambda\circ\tau(\{\rho_{i}^{1}\otimes\rho_{j}^{2}\})=\{\rho_{i}^{1}\otimes\rho_{j}^{2}\}$,
and, via the definition of
$\tau_{1}(M)\widetilde{\otimes}\tau_{2}(M)$, we have that
$\{\rho^{1}_{i}\otimes\rho^{2}_{j}\}\in\Lambda\circ\tau(M)$ for
all $i,j$. Thus, by the convexity of $\Lambda\circ\tau(M)$,
$\sigma\in\Lambda\circ\tau(M)$, which concludes the proof that
$\Lambda\circ\tau(M)=M$ (and  that $M$ is a polytope).
Consequently, $M$ is the desired
$S_{\rho}\subseteq\mathcal{S}(\mathcal{H})$.

\nd   If a given subset  $C \subseteq \mathcal{S}(\mathcal{H})$ then
all $\rho\in C$ are separable. $\mathcal{S}(\mathcal{H})$ is, by
definition, a convex set. Let us see that it is invariant under
$\Lambda\circ\tau$. First of all, we know that
$\mathcal{S}(\mathcal{H})$ is formed by all possible convex
combinations of products of the form $\rho_{1}\otimes\rho_{2}$, with
$\rho_{1}\in\mathcal{C}_{1}$ and $\rho_{2}\in\mathcal{C}_{2}$. But
for each one of these tensor products,
$\Lambda\circ\tau(\{\rho_{1}\otimes\rho_{2}\})=\{\rho_{1}\otimes\rho_{2}\}$,
and it is easy to see that they belong to
$\Lambda\circ\tau(\mathcal{S}(\mathcal{H}))$. Since this is a convex
set, all its convex combinations belong to it. Thus,  we conclude
that

\begin{equation}\label{e:S(H)isaCSS}
\Lambda\circ\tau(\mathcal{S}(\mathcal{H}))=\mathcal{S}(\mathcal{H}).
\end{equation}

\noindent This shows that for every $C\subseteq
\mathcal{S}(\mathcal{H})$ we can find an invariant convex subset
which is $\mathcal{S}(\mathcal{H})$ itself.

\vskip 3mm \nd Note here that there are cases in which   the set
$C\subseteq \mathcal{S}(\mathcal{H})$  may be a proper subset
(this is the case, for example, of product states) or a polytope
when we consider separable but non-product states. We remember at
this point the structural concept described by a definition of
section II.A. Consider $C\in\mathcal{L}_{\mathcal{C}}$ such that
there exists a given $\rho\in C$ with  $\rho$ {\it nonseparable}.
Now, $\Lambda\circ\tau(S)\subseteq \mathcal{S}(\mathcal{H})$ for
all $S\in\mathcal{L}_{\mathcal{C}}$. Then, it could never happen
that there exists $S\in\mathcal{L}_{\mathcal{C}}$ such that
$C\subseteq S$ and $\Lambda\circ\tau(S)=S$.
\end{proof}

\nd    From the last proposition we  derive our separability
criterium in terms of properties of convex sets that are
polytopes:

\begin{widetext}
\fbox{\parbox{6.0in}{
\begin{prop}\label{p:our criteria}
$\rho\in\mathcal{S}(\mathcal{H})$ if and only if there exists a
polytope $S_{\rho}$ such that $\rho\in S_{\rho}$ and
$\Lambda\circ\tau(S_{\rho})=S_{\rho}$.
\end{prop}}}
\end{widetext}

\nd    In Figure \ref{f:Politope} we display a geometric
representation of the polytope $S_{\rho}$ for a separable state.
We see that the function $\Lambda\circ\tau$ is sensible to
entanglement if applied to convex subsets of $\mathcal{C}$.
Looking at (\ref{e:lambdaenunrho}), it is also clear that
$\Lambda\circ\tau$ is a generalization of $\Omega$ to convex
subsets of $\mathcal{C}$.

\vskip 3mm \nd With this extension, Proposition \ref{p:our criteria}
asserts that

\begin{widetext}
\fbox{\parbox{5.0in}{ A state is separable if and only if it
belongs to an invariant polytope of $\Lambda\circ\tau$.
Separability entails membership in a special kind of convex set.}}
\end{widetext}

\vskip 3mm
 \nd  Clearly, starting from Proposition
\ref{p:our criteria} we can derive the family of functions of the
form (\ref{e:family}). Why? Because if we restrict the function
$\Lambda\circ\tau$ to convex sets formed by only one density matrix
we  obtain Equation (\ref{e:lambdaenunrho}) entailing that, if one
knows that $\Lambda\circ\tau$ is sensible to entanglement via
Proposition \ref{p:our criteria}, it is natural to regard  the norm
of the difference between $\rho$ and $\rho^{A}\otimes\rho^{B}$ as an
entanglement measure's candidate. Our set-theory approach becomes
then an {\it a posteriori} argument that in a sense ``explains'' the
MS measure.

\vskip 3mm \nd Let it be understood that we are restricting
$\Lambda\circ\tau$ to one-element sets $\{\rho\}$. With some abuse
of notation (which consists in avoiding the use of the keys
$\{\cdots\}$) we write

\begin{equation}\label{e:abuse}
\Lambda\circ\tau(\rho):=\rho^{A}\otimes\rho^{B}=\Omega(\rho)
\end{equation}

\begin{figure*}
\begin{center}
\includegraphics[width=8cm]{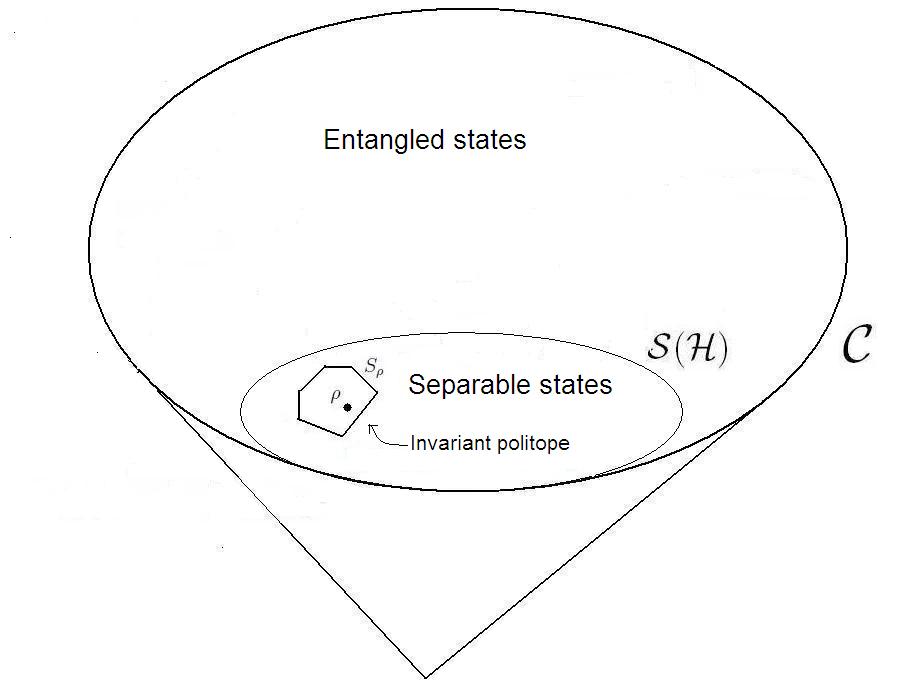}
\caption{\label{f:Politope}\small{Geometric representation of the
invariant polytope which satisfies
$\Lambda\circ\tau(S_{\rho})=S_{\rho}$ and $\rho\in S_{\rho}$.
$\rho$ is separable if and only if there exists such a polytope.}}
\end{center}
\end{figure*}

\section{Generalized product states}\label{s:Implications}

\nd We delve here into an interesting analogy. Denote the set of
product states by $\mathcal{S}_{0}(\mathcal{H}).$ Restricting
(\ref{e:abuse}) to product states we have

\begin{equation}\label{e:criteriaproduct}
\rho\in\mathcal{S}_{0}(\mathcal{H})\Leftrightarrow
\Lambda\circ\tau(\rho)=\rho\,\,\,(\Leftrightarrow\Omega(\rho)=\rho)
\end{equation}

From the discussion of the last section it is clear that our
criterium is analogous to (\ref{e:criteriaproduct}), being a
generalization of it to convex subsets of $\mathcal{C}$ because  we
have

\begin{equation}\label{e:exclusiveproperty}
\rho\in\mathcal{S}(\mathcal{H})\Leftrightarrow
\Lambda\circ\tau(S_{\rho})=S_{\rho},
\end{equation}
\noindent with $\rho\in S_{\rho}$. Accordingly, we are in some sense
generalizing a property of product states to any arbitrary separable
state. As $\Lambda\circ\tau$ generally transforms any convex set
into a different convex subset of $\mathcal{S}(\mathcal{H})$,
(\ref{e:exclusiveproperty}) constitutes a geometrical property,
characteristic of  separable states. Thus, we advance here a
``convex set" generalization of the notion of product state.

\begin{widetext}
\fbox{\parbox{6.0in}{
\begin{definition}
A convex subset $C\subseteq\mathcal{C}$ such that
$\Lambda\circ\tau(C)=C$ is called a \emph{convex separable subset}
(CSS) of $\mathcal{C}$.
\end{definition}}}
\end{widetext}

\vskip 3mm \nd    Due to the arguments given above, product states
are limit cases of convex separable subsets (they constitute the
special case when the CSS has only one point).

\vskip 3mm \nd  An interesting open problem would then be that of
looking for convex separable subsets of the function
$\Lambda\circ\tau$. Looking at (\ref{e:S(H)isaCSS}), we find that
$\mathcal{S}(\mathcal{H})$ is a CSS (and indeed, the largest one).
In this sense, CSS may be considered as small ``copies" of
$\mathcal{S}(\mathcal{H})$.

\nd    In general, convex subsets of $\mathcal{C}$ may be
considered as probability spaces by themselves, because they are
closed under convex combination of states. Thus, {\it CSS are
probability spaces inside} $\mathcal{S}(\mathcal{H})$ which are
left invariant under the action of $\Lambda\circ\tau$ (and so,
they have the same invariance property). The fact that
$\mathcal{S}(\mathcal{H})$ is a CSS also tells us that the convex
separable subsets can be more general sets and not necessarily
just polytopes (because $\mathcal{S}(\mathcal{H})$ is not a
polytope). Indeed, we may ask for ways to characterize the set of
all convex separable subsets (which we denote by
$\daleth(\mathcal{C})$) by looking at the following property of
$\Omega$. If $\rho$ is an arbitrary density matrix, then

\begin{eqnarray}
&\Omega^{2}(\rho)=\Omega(\Omega(\rho))=\Omega(\rho^{A}\otimes\rho^{B})=&\nonumber\\
&\rho^{A}\otimes\rho^{B}=\Omega(\rho)\nonumber\\
\end{eqnarray}
\noindent or, in other words,

\begin{equation}\label{e:omegacuadrado}
\Omega^{2}=\Omega.
\end{equation}

\nd For $\Lambda\circ\tau$ and an arbitrary convex subset $C$ one
has

\begin{eqnarray}
&\Lambda\circ\tau(C)=\Lambda(\tau_{1}(C),\tau_{2}(C))=&\nonumber\\
&Conv(\tau_{1}(C)\widetilde{\otimes}\tau_{2}(C))&.
\end{eqnarray}
If we apply $\Lambda\circ\tau$ again, we will find (with arguments
expounded in the preceding section, see \ref{subirbajar}) that
$Conv(\tau_{1}(C)\widetilde{\otimes}\tau_{2}(C))$ is a CSS. This, in
turn, entails that

\begin{equation}\label{e:lambdataucuadrado}
(\Lambda\circ\tau)^{2}=\Lambda\circ\tau.
\end{equation}
\noindent Consequently, our generalization of $\Omega$ satisfies an
equality equivalent to (\ref{e:omegacuadrado}). This fact can be
gainfully  used to characterize $\daleth(\mathcal{C})$ as

\begin{equation}\label{e:characterizationofaleth}
\daleth(\mathcal{C})=\{\Lambda\circ\tau(C)\,\,|\,\,C\subseteq\mathcal{C}\},
\end{equation}
\noindent because, if $C$ is a CSS, it is equal to
$\Lambda\circ\tau(C)$, and thus we face one inclusion. The other
inclusion comes from the fact that, for an arbitrary
$C\subseteq\mathcal{C}$, (\ref{e:lambdataucuadrado}) implies that
$\Lambda\circ\tau(C)$ belongs to $\daleth(\mathcal{C})$. Equation
\ref{e:characterizationofaleth} simply asserts that
$\daleth(\mathcal{C})$ equals the image of
$\mathcal{L}_{\mathcal{C}}$ under $\Lambda\circ\tau$.

\nd Now we see that  while in equation (\ref{e:family}), the ``core"
was the function $\rho-\Lambda\circ\tau(\rho)$, now we have a new
core

\begin{equation}
\Lambda\circ\tau(C)\setminus C,
\end{equation}

\noindent where ``$\backslash$'' stands for set-theoretical
difference, and we can try to measure the difference between $C$ and
its variation under $\Lambda\circ\tau$ in different ways. We will
have a CSS if $C$ and $\Lambda\circ\tau(C)$ coincide. \vskip 3mm \nd
A possibility for measuring how different are $C$ and
$\Lambda\circ\tau(C)$ would entail looking for a generalization of,
for example, the relative entropy, which for a density matrix reads

\begin{equation}
S(\rho\|\sigma)=-\mbox{tr}(\rho \mbox{log}(\sigma))-S(\rho),
\end{equation}

\noindent where $S(\rho):=-\mbox{tr}(\rho \mbox{log}(\rho))$. Remark
that the relative entropy concept has been used as a unifying
approach for quantum and classical correlations \cite{UnifiedView}.
When applied to convex subsets $C$ and $C'$ of $\mathcal{C}$, we are
now conjecturing that

\begin{equation}
S(C\|C'):=\inf_{\rho\in C,\sigma\in C'} S(\rho\|\sigma),
\end{equation}
\noindent and use this conjecture to define

\begin{equation}\label{e:relativeentropylambdatau}
\widetilde{F}(C):=S(\Lambda\circ\tau(C)\|C).
\end{equation}
\nd    $\widetilde{F}(C)$ clearly vanishes when
$\Lambda\circ\tau(C)=C$, and in general, when
$\Lambda\circ\tau(C)\cap C\neq\emptyset$. This last condition
implies (in particular) that there are separable states which belong
to $C$. We are free to use any divergence (or distance) instead of
the relative entropy for the purpose of  measuring the difference
between $C$ and $\Lambda\circ\tau(C)$ by making a similar
construction.

\vskip 3mm \nd    Let us now study the segment joining $\rho$ and
$\Lambda\circ\tau(\rho)$. This segment is given by

\begin{equation}\label{e:line}
\mathrm{L}_{\rho}=\{x\rho+(1-x)\Lambda\circ\tau(\rho)\,\,|\,\,x\in[0,1]\}.
\end{equation}
\noindent If $\rho$ is separable, using i) the polytope
$S_{\rho}\subseteq \mathcal{S}(\mathcal{H})$ of proposition
\ref{subirbajar}, ii) that $\rho$ and $\Lambda\circ\tau(\rho)$
belong to $S_{\rho}$,  and iii) that $S_{\rho}$ is convex, we have

\begin{prop}\label{p:segment}
$\\\mathrm{L}_{\rho}\subseteq S_{\rho}\subseteq
\mathcal{S}(\mathcal{H}).$
\end{prop}
\nd Coming back again to the  demonstration of Proposition
\ref{subirbajar}, and considering that the decomposition of a
separable state as a convex combination of product states is not
unique, we conclude that the invariant polytope is not unique.
However, from the above proposition it is obvious that

\begin{equation}\label{e:lineintersection}
\mathrm{L}_{\rho}\subseteq \cap\{C \,|\,\Lambda\circ\tau(C)=C \,
\mbox{and} \, \rho\in C\}
\end{equation}

If there exists at least one nonseparable state in the segment
joining $\rho$ and $\Lambda\circ\tau(\rho)$, then $\rho$ cannot be a
separable state.  This is a consequence of the convexity of
$\mathcal{S}(\mathcal{H})$, but also follows from
(\ref{e:lineintersection}). Is this fact  of  advantage for deciding
on the separability of a given state? Indeed it is, if we use it in
the following way. Given $\rho$, we parameterize the line segment
between $\rho$ and $\rho^{A}\otimes\rho^{B}$ as in proposition
\ref{p:segment}. Afterwards,  we apply this to all the points in the
segment. If one  finds a nonseparable state in the segment we
conclude that $\rho$ is nonseparable.

\vskip 3mm \nd We consider now the action of the group of unitary
local transformations of the form $U=U^{1}\otimes U^{2}$ on the
invariant polytope, where $U^{1,2}\in U^{\mathcal{K}^{1,2}}$. If
$\rho=\sum_{i}p_{i}\rho_{i}^{A}\otimes\rho_{i}^{B}$ is a separable
state, then this action will be given by

\begin{equation}\label{e:unitarytransfrho}
U\rho U^{\dagger}=\sum_{i}p_{i}U^{1}\rho_{i}^{A}U^{1\dagger}\otimes
U^{2}\rho_{i}^{B}U^{2\dagger}.
\end{equation}
We can prove that

\begin{prop}\label{p:unitaryprop}
If $\rho\in\mathcal{S}(\mathcal{H})$ and $\mathrm{P}_{\rho}$ is an
invariant polytope (as the one in the demonstration of proposition
\ref{subirbajar}), then $U\mathrm{P}_{\rho}U^{\dagger}$ is an
invariant polytope for $U\rho U^{\dagger}$.
\end{prop}
\begin{proof}
If $\rho=\sum_{i}p_{i}\rho_{i}^{A}\otimes\rho_{i}^{B}$, then an
invariant polytope is given by

\begin{equation}
\mathrm{P}_{\rho}=\{\sum_{i,j}\lambda_{ij}\rho_{i}^{A}\otimes\rho_{j}^{B}\,|\,\sum_{i,j}\lambda_{ij}=1\,,\,\lambda_{ij}\geq
0\}.
\end{equation}

\noindent Because of the linearity of $U$, it is easy to see that
$\mathrm{P}_{\rho}$ is transformed into

\begin{eqnarray}
U\mathrm{P}_{\rho}U^{\dagger}=\{\sum_{i,j}\lambda_{ij}U^{1}\rho_{i}^{A}U^{2\dagger}\otimes
U^{2}\rho_{j}^{B}U^{2\dagger}\,|\,\nonumber\\
\sum_{i,j}\lambda_{ij}=1\,,\,\lambda_{ij}\geq 0\},
\end{eqnarray}

and as $\rho$ is transformed as equation
(\ref{e:unitarytransfrho}), then $U\mathrm{P}_{\rho}U^{\dagger}$
is an invariant polytope.
\end{proof}
\nd   The last proposition shows how invariant polytopes are
transformed under unitary local transformations. As
$\mathcal{S}(\mathcal{H})$ is invariant under these
transformations, we see that they transform invariant polytopes
into other invariant polytopes. Notice that \ref{p:unitaryprop}
implies (for invariant polytopes) that under an arbitrary local
transformation $U$

\begin{equation}
\Lambda\circ\tau(U\mathrm{P}_{\rho}U^{\dagger})=U\mathrm{P}_{\rho}U^{\dagger}=U(\Lambda\circ\tau(\mathrm{P}_{\rho}))U^{\dagger},
\end{equation}
\noindent which reveals an interesting  symmetry property of
$\Lambda\circ\tau$.

\section{Discussion}
\subsection{A conceptual analogy}\label{s:moreconceptual}

\noindent For clarity's sake we condense here in a more conceptual
fashion some of the technical implications of the foregoing
sections via  appeal to a
comparison with the separability-notion for pure states. Its characterization in the bipartite is simple.
 $\rho=|\psi\rangle\langle\psi|$ will be separable if and only
if it is a product of pure reduced states, i.e., if and only if
there exist $|\phi_{2}\rangle\in\mathcal{H}_{1}$ and
$|\phi_{2}\rangle\in\mathcal{H}_{2}$ such that
$|\psi\rangle=|\phi_{1}\rangle\otimes|\phi_{2}\rangle$. In
mathematical terms (take care of not to be  confused by equation
(\ref{e:criteriaproduct}))

\begin{eqnarray}\label{e:pureseparable}
|\psi\rangle\langle\psi|\in\mathcal{S}(\mathcal{H})\Leftrightarrow\Lambda\circ\tau(|\psi\rangle\langle\psi|)=|\psi\rangle\langle\psi|\nonumber\\
(\Leftrightarrow\Omega(|\psi\rangle\langle\psi|)=|\psi\rangle\langle\psi|).
\end{eqnarray}
\noindent It is well known that the case of mixed states is much
more complicated than that of pure ones. But we may still ask if
it is possible to develop a similar line of reasoning for mixed
states. The existence of such a construction would allow for a
more transparent  view of the  entanglement of mixed states (and
thus for \emph{all} states, generalizing (\ref{e:pureseparable})).
The results and constructions presented in previous sections of
this article indicate that a structure similar to that of equation
(\ref{e:pureseparable}) can indeed be constructed. \vskip 3mm \nd
This fact makes for  a remarkable analogy, unknown in the
literature, whose  explanation is as follows. We showed in section
\ref{s:Implications} that the function $\Lambda\circ\tau$
(introduced in section \ref{s:New Criteria}) is a suitable
extension to convex subsets of the function $\Omega$ (look at
Definition \ref{e:assignment}). We also introduced the
physical-informational notion of CSS, an informational invariant
convex subset, i.e., a set whose information can be recovered
using the sets of its corresponding reduced states. In this sense,
\emph{they are informational invariants}. As shown in section
\ref{s:Implications}, \emph{they are a suitable generalization of
the notion of product state to all convex subsets of
$\mathcal{C}$}.

\noindent Thus, as happens in the pure state case, we have
developed a generalization which asserts that an arbitrary state
is separable if and only if it is an element of an informational
invariant that we have called CSS. Our math-constructions and
entanglement criteria (linked to the SM measure) highlight the
non-trivial result that the structure found for the pure states
case can be properly generalized to  arbitrary states,  a clear
 physical simplification.

\noindent But the analogy/generalization does not stops here. We
can develop still a new analogy/generalization,  not contained in
the precedent  sections. It is well known that another equivalent
condition for separability of pure states may be given using von
Neuman's entropy, which reads: $\rho=|\psi\rangle\langle\psi|$ is
separable if and only if the von Neuman's entropy of its reduced
states attains its minimum possible value (zero). In mathematical
terms

\begin{eqnarray}\label{e:pureseparableshanon}
|\psi\rangle\langle\psi|\in\mathcal{S}(\mathcal{H})\Leftrightarrow
S_{vN}(\rho^{A})=0\nonumber\\
\mbox{and}\,\, S_{vN}(\rho^{B})=0,
\end{eqnarray}
\noindent where $\rho^{A}$ and $\rho^{B}$ are the reduced states
of $|\psi\rangle\langle\psi|$ and $S_{vN}(\cdot)$ is the well
known von Neuman's entropy functional, defined by

\begin{equation}
S_{vN}=-\mbox{tr}(\rho\mbox{ln}(\rho)).
\end{equation}
\noindent Can we concoct something similar for mixed states?
Caratheodory's theorem (for finite dimensions) grants that any
separable state admits a finite convex decomposition in terms of
pure product states. In mathematical terms, this means that there
exists pure states
$|\varphi_{i}\rangle\langle\varphi_{i}|\in\mathcal{C}_{1}$,
$|\phi_{i}\rangle\langle\phi_{i}|\in\mathcal{C}_{2}$ and a finite
collection of convex coefficients $\lambda_{i}$ such that

\begin{equation}
\rho\in\mathcal{S}(\mathcal{H})\Leftrightarrow
\sum_{i}\lambda_{i}(|\varphi_{i}\rangle\langle\varphi_{i}|)\otimes(|\phi_{i}\rangle\langle\phi_{i}|).
\end{equation}
\noindent It is easy to show that this decomposition combined with
our separability criteria \ref{p:our criteria} (look at the
demonstration of it) implies that there exists a polytope, call it
$P_{pure}$, whose vertices are just products of pure states. This
implies that if we now compute the infimum of the von Neuman
entropy evaluated on the elements of $P_{pure}$ we will obtain its
minimum value, because as it is well known, von Neuman entropy
attains its minimum value for such states. In other words,

\begin{widetext}
\begin{equation}
\inf\{S_{vN}(\rho)\,|\,\rho\in
P_{pure}\}=\min\{S_{vN}(\rho)\,|\,\rho\in P_{pure}\}=0.
\end{equation}
\end{widetext}
\noindent Thus, the analogy advanced in this section is more than
a simple coincidence or mathematical artifice, because in accord
with Eq. (\ref{e:pureseparableshanon}), we now have that for
\emph{any state (pure or mixed)},

\begin{eqnarray}
\rho\in\mathcal{S}(\mathcal{H})\Leftrightarrow \exists
P_{pure},\,\,\mbox{such that}\nonumber\\
\min\{S_{vN}(\rho)\,|\,\rho\in P_{pure}\}=0,
\end{eqnarray}
\noindent where $P_{pure}$ represents a polytope whose vertices
are products of pure states. Thus, we can sum up some of the
results of this article by just using the following words:

\begin{widetext}
\fbox{\parbox{6.0in}{
\begin{prop}\label{p:our criteria}
$\rho$ is a separable state $\Leftrightarrow$ it belongs to a CSS
(i.e., a convex subset which generalizes product states and is
invariant under the function defined by equation
(\ref{e:assignment})) $\Leftrightarrow$ it belongs to a  CSS on
which the von Neuman's entropy reaches its minimum value.
\end{prop}}}
\end{widetext}

\noindent The analogy with the pure case is not only  clear and
suggestive. It may also  provide some geometric flavor to the
separability problem. It is indeed a generalization which includes
the pure case as a special one. Interestingly enough, as shown
in section \ref{s:New Criteria}, it is strongly linked to the
$S-M$ measure.

\subsection{Final conclusions}\label{s:Conclusions}

\nd We have advanced here an abstract criterium of separability and
showed that it is closely  connected to the  extant entanglement
measures. We ascertained also that the function $\Lambda\circ\tau$
is a generalization of the map $\rho\mapsto \rho^{A}\otimes\rho^{B}$
to convex subsets of $\mathcal{C}$. Indeed, we showed that
$\Lambda\circ\tau$ generalizes  to convex sets properties of
invariant product states of the map $\rho^{A}\otimes\rho^{B}$.
\vskip 3mm \nd Denoting by ``CSS" the invariant subsets of
$\mathcal{C}$, a procedure was delineated that generalizes product
states to more general convex sets. This could be useful for the
study of new separability criteria based on more general convex
subsets of $\mathcal{C}$ and  disposes of the obligation of
concentrating attention just on points (density matrixes). By
itself, the criterium \ref{p:our criteria} also sheds some light
into aspects of the geometric properties of separable states.


\appendix
\section{Notations for basic math concepts used in the text}

\begin{enumerate}

\item A function is surjective (onto) if every possible image is mapped
to by at least one argument. In other words, every element in the
codomain has non-empty preimage. Equivalently, a function is
surjective if its image is equal to its codomain. A surjective
function is a surjection.

\item Let $S$ be a vector space over the real numbers, or, more generally, some ordered field.
A set $C$ in $S$ is said to be convex if, for all $x$ and $y$ in $C$
and all $t$ in the interval $[0,1]$, the point $(1 - − t ) x + t
y$ is in $C$. That is, every point on the line segment connecting
$x$ and $y$ belongs to $C$. This entails that any convex set in a
real or complex topological vector space is path-connected.

\item Every subset $Q$ of a vector space is contained within a
smallest convex set (called the convex hull of $Q$), namely the
intersection of all convex sets containing $Q$,

\item  Suppose that $K$ is a field (for example, the real numbers)
and $V$ is a vector space over $K$. If $v_1,\ldots,v_n$ are vectors
and $a_1,\ldots,a_n$ are scalars, then the linear combination of
those vectors with those scalars as coefficients is, of course,
$\sum_{i=1}^n\,a_i\,v_i$.  By restricting the coefficients used in
linear combinations, one can define the related concepts of affine
combination, conical combination, and convex combination, together
with the associated notions of sets closed under these operations.
If $\sum_{i=1}^n\,a_i=1,$ we have an affine combination, its span
being an affine subspace while  the model space is an hyperplane. If
$a_i \ge 0,$ we have instead a conical combination, a convex cone
and a quadrant, respectively. Finally, if $a_i \ge 0$ plus
$\sum_{i=1}^n a_i=1$  we have now a convex combination, a convex set
and a simplex, respectively.

\end{enumerate}

\vskip1truecm

\noindent {\bf Acknowledgements} \noindent This work was partially
supported by the following grants: .

\end{document}